%% file: dilutionpaper.tex
\documentclass[a4paper,accepted=2018-06-20]{quantumarticle}
\pdfoutput=1
\usepackage[utf8]{inputenc}

\usepackage[T1]{fontenc}
\usepackage{microtype}

\usepackage{enumitem}

\usepackage{amsthm,thm-restate}
\usepackage{amsmath,amsfonts}
\usepackage{IEEEtrantools}
\usepackage{mathtools}
\usepackage{mleftright}
\usepackage{braket}

\usepackage[
	style=phys,
	biblabel=brackets,
	autolang=hyphen,
	backend=biber,
	sorting=none,
	eprint=true,
	doi=true,
]{biblatex}
\usepackage{csquotes}

\input{h_macros.tex}

\date{June 16, 2018}

\begin{document}
\title{Device-independent randomness generation with sublinear shared quantum resources}

\author{C\'edric Bamps}
\author{Serge Massar}
\author{Stefano Pironio}
\affiliation{Laboratoire d'Information Quantique, CP 224, Universit\'e libre de Bruxelles (ULB), 1050 Brussels, Belgium}

\begin{abstract}
In quantum cryptography, device-independent (DI) protocols can be certified secure without requiring assumptions about the inner workings of the devices used to perform the protocol.
In order to display nonlocality, which is an essential feature in DI protocols, the device must consist of at least two separate components sharing entanglement.
This raises a fundamental question: how much entanglement is needed to run such DI protocols?
We present a two-device protocol for DI random number generation (DIRNG) which produces approximately $\kround$ bits of randomness starting from $\kround$ pairs of arbitrarily weakly entangled qubits.
We also consider a variant of the protocol where $\kebit$ singlet states are diluted into $\kround$ partially entangled states before performing the first protocol, and show that the number $\kebit$ of singlet states need only scale sublinearly with the number $\kround$ of random bits produced.
Operationally, this leads to a DIRNG protocol between distant laboratories that requires only a sublinear amount of quantum communication to prepare the devices.
\end{abstract}

\maketitle

\section{Introduction}
A quantum random number generation (RNG) protocol is device-independent (DI) if its output can be guaranteed to be random with respect to any adversary on the sole basis of certain minimal assumptions, such as the validity of quantum physics and the existence of secure physical locations \cite{AM16}.
The internal workings of the devices, however, do not need to be trusted.

Device-independence is made possible by exploiting the violation of a Bell inequality \cite{BCP+14}, which certifies the random nature of quantum measurement outcomes.
As a result, DIRNG protocols necessarily consume two fundamental resources: entangled states shared across separated devices and an initial public random seed that is uncorrelated to the devices and used to determine the random measurements performed on the entangled states.
Out of these two resources, a DIRNG protocol produces $\kround$ private random bits.

The initial random seed that is consumed can be of extremely low quantity or quality.
Indeed, $\kround$ private random bits can be produced starting from an initial string of uniform bits whose required length has gradually been reduced in a series of works \cite{PAM+10,VV12,CY13,MS14a}, culminating in the result that only a constant, i.e., independent of the output length $\kround$, amount of initial uniform random bits are required \cite{CY13,MS14a}.
Furthermore, the initial seed does not necessarily need to consist of uniform random bits, as it possible to design DIRNG protocols consuming an arbitrarily weak random seed characterized only by its total min-entropy~\cite{CSW14}.

What about entanglement, the second fundamental resource that is consumed in any DIRNG protocol?
This quantum resource usually consists of $\kebit$ copies $\ket{\psi}^{\otimes \kebit}$ of some bipartite entangled state $\ket{\psi}$ shared between two separated devices $\sA$ and $\sB$ that can be prevented at will from interacting with one another.
Though DIRNG protocols involve a single user, it is useful for exposition purposes to view these two devices as being operated by two agents, Alice and Bob, in two remote sublaboratories.
The $\kebit$ copies $\ket{\psi}^{\otimes \kebit}$ can either be stored prior to the start of the protocol inside quantum memories in Alice's and Bob's sublaboratories, or each copy $\ket{\psi}$ can be produced individually during each execution round of the protocol, say by a source located between Alice and Bob.

All existing protocols consume at best a linear amount $\kebit = \bigomega(\kround)$ of such shared entangled states $\ket{\psi}$, as they operate by separately measuring (in sequence or in parallel) each of these $\kebit$ copies, with each separate measurement yielding at most a constant amount of random bits.
Furthermore, the states $\ket{\psi}$ are typically highly entangled states---the prototypical example of a DIRNG protocol involves the measurement of $\kround$ maximally entangled two-qubit states $\ket{\phi^+}$, from each of which roughly $1$ bit of randomness can be certified using the CHSH inequality \cite{PAM+10}.

We will show that the consumption of entangled resources can be dramatically improved qualitatively and quantitatively.
First, we show---by analogy with the fact that the initial random seed does not need to consist of uniform bits---that highly entangled states are not necessary for DIRNG: instead of using $\kround$ copies of maximally entangled two-qubit pairs $\ket{\phi^+}$, $\kround$ random bits can be produced from $\kround$ copies of any partially entangled two-qubit pair $\ket{\psi_\theta} = \cos\theta \ket{00} + \sin\theta \ket{11}$ with $0<\theta\leq \pi/4$ (see Theorem~\ref{thm:minentropy} and Corollary~\ref{cor:rng-partial}).

We then turn this statement concerning the quality of the shared entangled resources into a quantitative statement about the amount of entanglement that needs to be consumed in a DIRNG protocol.
The $\kround$ copies of the partially entangled state $\ket{\psi_\theta}$ correspond to a total of $\kround S(\theta)$ ebits where $S(\theta) = h_2(\sin^2 \theta)$ is the entropy of entanglement of $\ket{\psi_\theta}$ expressed in terms of the binary entropy $h_2$.
Since $S(\theta)$ can be made arbitrarily low by considering sufficiently small values of $\theta$, the above result seems to suggest that the total amount $\kround S(\theta)$ of entanglement consumed can also be made arbitrarily small as a function of $\kround$ by considering sufficiently fast decreasing values for $\theta = \theta(\kround)$.
However, if it is true that for any given $\theta$, one can produce $\kround$ random bits from $\kround$ copies of $\ket{\psi_\theta}$ for any $\kround$ sufficiently large, the dependency between $\theta$ and $\kround$ cannot be chosen arbitrarily.
This essentially originates from the fact that as $\theta \to 0$ the robustness to noise of the corresponding states $\ket{\psi_\theta}$, which become less and less entangled, decreases and must be compensated by increasing the number $\kround$ of copies of the states $\ket{\psi_\theta}$ to improve the estimation phase of the protocol.
There is thus a tradeoff between $\theta$ and $\kround$, which we show can nevertheless result in a total amount of entanglement $\kround S(\theta) = \bigomega(\kround^k\log \kround)$ with $7/8<k<1$ (see Corollary~\ref{cor:ebit}).
This amount of entanglement is \emph{sublinear} in the number $\kround$ of output random bits, fundamentally improving over existing protocols for which the entanglement consumption is at best linear.

Though the protocol that we introduce consumes a sublinear amount of entanglement, it still requires a linear number of shared quantum resources in the form of $\kround$ copies of the two-qubit entangled states $\ket{\psi_\theta}$.
These shared entangled states must be established through some quantum communication between Alice's and Bob's sublaboratories, either during the protocol itself or prior to the protocol, and will thus require the exchange of $\kround$ qubits.
Since this quantum communication will typically be costly (for instance because of high losses in the communication channel), it represents a measure of the use of shared quantum resources which is more operational and better motivated than the entropy of entanglement.
From this perspective, however, our first protocol is not fundamentally different from existing protocols that also involve the exchange of $\kround$ qubits to produce $\kround$ random bits.

This leads us to consider a slight modification of our protocol in which Alice and Bob initially share $\kebit$ maximally entangled two-qubit states $\ket{\phi^+}$, which can be established through the exchange of $\kebit$ qubits.
These singlets are then transformed by entanglement dilution \cite{BBPS96} into roughly $\kround = S(\theta)/\kebit$ copies of $\ket{\psi_\theta}$ states through local operations and classical communication (LOCC), which are then used in our regular protocol.

However entanglement dilution is only noiseless asymptotically, in the limit of an infinite number of copies $\kebit \to \infty$.
For finite $\kebit$, entanglement dilution is inherently noisy.
As our protocol is increasingly sensitive to noise as the degree of entanglement of the states $\theta$ tends to $0$, it is not a priori obvious that combining randomness generation with entanglement dilution will work.

Nevertheless we show that such a two-step protocol works even though the entanglement dilution slightly degrades the tradeoff between $\theta$ and $\kround$.
Specifically we exhibit a protocol that can get $\kround$ output random bits starting from a sublinear number $\kebit = \kround S(\theta) = \bigomega(\kround^{k'}\log \kround)$ of initial copies of $\ket{\phi^+}$ states, with $7/8<k'<1$.
This represents a quantitative improvement of the use of quantum resources with respect to all existing protocols, analogous to the fact that a DIRNG protocol needs only a sublinear amount of uniform random bits.

The starting point of our work is the work \cite{AMP12} wherein a family of variants of the CHSH inequality, the tilted-CHSH inequalities, are introduced, which seem particularly suited to generate randomness from weakly entangled qubit states.
Indeed, it was shown in \cite{AMP12} that maximal violation of a tilted CHSH inequality certifies one bit of randomness and can be achieved by entangled two-dimensional systems with arbitrarily little entanglement%
\footnote{%
Note that not all weakly entangled states can be used for device-independent randomness generation: for instance there is a regime of visibility in which noisy singlet states (so-called Werner states) are entangled but incapable of displaying nonlocality, and hence also incapable of displaying randomness \cite{Wer89}.
}
This was later extended to show that by using sequential measurements, a single pair of entangled qubits in a pure state could certify an arbitrary amount of randomness \cite{CJA+17}.
However neither of these works presented a protocol, including an estimation phase and security analysis taking into account non-maximal violation, for device independent randomness generation.
In fact the results of \cite{AMP12,CJA+17} do not by themselves imply the existence of such a protocol.

We now recall the tilted-CHSH expressions of \cite{AMP12}, whose properties of randomness certification in weakly entangled states will play a central role in our protocol.

\subsection*{Tilted-CHSH game}
The tilted-CHSH expressions $I_1^\beta$ are a family of Bell expressions introduced in \cite{AMP12} and parameterized by a tilting parameter $\beta\in\coint{0,2}$.
We start by reformulating $I_1^\beta$ as a nonlocal game, expressed in terms of a predicate function $V\in \{0,1\}$.
This will put us in the right conditions to apply the entropy accumulation theorem of \cite{DFR16} following \cite{AFVR16}.
In this reformulation,
Alice is given input $\sa \in \{0,1\}$ and Bob input $\sbb \in \{0,1,2\}$ according to the joint distribution
\begin{equation}
\label{eq:game-distribution}
	p(\sa,\sbb) = \begin{cases}
		\frac{1}{4+\beta}      &  (\sa,\sbb) \in \{0,1\}^2 \eqp, \\
		\frac{\beta}{4+\beta}  &  (\sa,\sbb) = (0,2) \eqp, \\
		0                      &  \text{otherwise.}
	\end{cases}
\end{equation}
Alice and Bob then provide one answer each, $(\oa,\ob) \in \{0,1\}^2$ respectively, and the game is won if the following predicate function $V(\oa,\ob,\sa,\sbb) \in \{0,1\}$ returns $1$:
\begin{equation}
\label{eq:game-predicate}
	V(\oa,\ob,\sa,\sbb) = \begin{cases}
		1  & (\sa,\sbb) \in \{0,1\}^2 \text{ and } \oa \oplus \ob = \sa \sbb \eqp, \\
		1  & (\sa,\sbb) = (0,2) \text{ and } \oa = 0 \eqp, \\
		0  & \text{otherwise.}
	\end{cases}
\end{equation}
Note that in our reformulation of the tilted-CHSH expression as a game, we have introduced for convenience a third setting for Bob ($y=2$) that is absent in the original tilted-CHSH expression.
This game can be understood as a convex combination between the CHSH game and a ``trivial'' game: the former's success criterion is $\oa \oplus \ob = \sa \sbb$ with input probabilities $p(\sa,\sbb) = 1/4$ for $(\sa,\sbb) \in \{0,1\}^2$, while the latter's success criterion is $\oa = 0$ with a deterministic input $(\sa,\sbb) = (0,2)$.
While Bob can tell the two games apart from his input $y$ thanks to the introduction of the third setting $y=2$, from Alice's point of view they are not distinguishable.
This makes the mixture of the CHSH game with the trivial game nontrivial.

Given the predicate function (\ref{eq:game-predicate}), it can easily be verified that the expected winning probability $\omega$ for the tilted-CHSH game is linked to the expectation value $\bar I_1^\beta$ of the tilted-CHSH expression through
\begin{align}
\omega &= \sum_{\mathclap{\oa,\ob,\sa,\sbb}} V(\oa,\ob,\sa,\sbb) p(\sa,\sbb) p(\oa,\ob \mid \sa,\sbb) \\
	&= \frac12 + \frac{1}{8+2\beta} \bar I_1^\beta \label{eq:game-ineq}\eqp,
\end{align}
where $p(\oa,\ob \mid \sa,\sbb)$ are the probabilities characterizing Alice and Bob's outputs.

Note that when $\sbb=2$, Bob's output does not affect the outcome of the game, and Bob is free to provide any output.
We expect that it should be possible to reformulate our results without introducing Bob's third setting, but we have found it simplest to proceed as above in order to follow closely the results of \cite{AFVR16}, where non-local games are used rather than Bell inequalities.

From the relation (\ref{eq:game-ineq}) between the tilted-CHSH game and the tilted-CHSH expression, it follows from the results of \cite{AMP12} that the winning probability $\omega$ goes up to $1/2 + (2+\beta)/(8+2\beta)$ for classical devices, and $1/2 + \sqrt{8+2\beta^2}/(8+2\beta) = \omega_\quant$ for quantum devices.
This quantum value $\omega_\quant$ is uniquely achieved (up to local transformations and up to Bob's measurement operator for $y=2$) by a pair of devices implementing certain local measurements on a two-qubit partially entangled state $\ket{\psi_\theta} = \cos\theta \ket{00} + \sin\theta \ket{11}$ with $\tan(2\theta) = \sqrt{2/\beta^2 - 1/2}$ \cite{AMP12}.
We call this optimal pair of devices the \emph{reference devices} for the tilted-CHSH game of tilting parameter $\beta$.
In the following we will sometimes use $\theta$ as the game parameter instead of $\beta$; it is always understood that they are linked by the above relation.

One important feature of the reference devices, as highlighted in \cite{AMP12}, is that, for any $0<\theta\leq \pi/4$, Alice's measurement when $\sa = 1$ returns a uniformly distributed outcome $\oa \in \{0,1\}$ uncorrelated with the environment, i.e., one bit of ideal randomness.
Thus by separately measuring $\kround$ copies of the partially entangled state $\ket{\psi_\theta} = \cos\theta \ket{00} + \sin\theta \ket{11}$ according to the reference measurements, one could in principle generate $\kround$ bits of randomness for any $0<\theta\leq \pi/4$.

However, the results of \cite{AMP12} do not immediately imply this claim because they only apply to a single use of a quantum system that is known to achieve the maximal winning probability $\omega_\quant$ of the tilted-CHSH game.
Thus one should first embed the tilted-CHSH game in a proper DIRNG protocol in which no assumptions are made beforehand about the quantum systems, but where the amount of randomness generated is instead estimated from their observed behavior.
This requires in particular a robust version of the results of \cite{AMP12}, i.e., an assessment of the randomness produced by quantum devices achieving a suboptimal winning probability $\omega<\omega_\quant$.
Indeed, even ideal devices are not expected to achieve the quantum maximum when they are used a finite number $\kround$ of times because of inherent statistical noise.
We now address this by introducing an explicit DIRNG protocol based on the tilted-CHSH inequalities and a robust security analysis based on the entropy accumulation theorem (EAT) \cite{DFR16,AFVR16,ADF+18} and the self-testing properties of the tilted-CHSH inequalities introduced in \cite{BP15}.

\section{DIRNG protocol based on the tilted-CHSH game}
\label{sec:protocol}
Our protocol consists of the following steps:
\begin{enumerate}[noitemsep]
\item
  Select values for the following parameters:
\begin{itemize}[nosep]
\item The game parameter $\beta \in \coint{0,2}$;
\item The number of measurement rounds $\kround$;
\item The expected fraction of test rounds $\gamma$;
\item A success threshold $\omega_\quant - \xi$.
\end{itemize}
\item \label{item:protocol:round}
  Let $i=1$.
  Choose $T_i \in \{0,1\}$ independently at random such that $\Pr[T_i = 1] = \gamma$.
  If $T_i = 1$, perform a game round: measure the devices with settings $(\Sa_i,\Sb_i)$, selected at random according to the distribution given in \eqref{eq:game-distribution}, record the output $(\Oa_i,\Ob_i)$ and compute $C_i = V(\Oa_i,\Ob_i,\Sa_i,\Sb_i)$ according to \eqref{eq:game-predicate}.
  If $T_i = 0$, perform a generation round: measure the devices with $(\Sa_i,\Sb_i) = (1,0)$, record the output $(\Oa_i,\Ob_i)$ and let $C_i = \bot$.
\item
  Repeat step~\ref*{item:protocol:round} for $i = 2, \dotsc, \kround$.
\item \label{item:protocol:estimation}
  Finally, if $\sum_{i : C_i = 1} 1 \ge \kround \gamma (\omega_\quant - \xi)$, the protocol succeeds.
  Otherwise, it aborts.
\end{enumerate}
An immediate application of Hoeffding's inequality \cite{Hoe63} produces an upper bound on the \emph{completeness error} for this protocol, that is, the probability that the ideal devices fail the protocol:
\begin{lemma}
\label{lem:completeness}
Using the reference devices in the $\kround$ rounds, the completeness error for the protocol is bounded by
\begin{equation}\label{eq:comperr}
	\epsilon_\compl = \exp \mleft( -2 \kround (\gamma \xi)^2 \mright) \eqp.
\end{equation}
\end{lemma}

\subsection{Soundness of the protocol}
We now establish the soundness of our protocol, that is, its ability to produce a positive amount of randomness with high probability given that the protocol did not abort.
The security of this protocol rests on three standard assumptions in the DI setting: that the devices and their environment obey the laws of quantum mechanics, that the random seed used to select inputs is independent from the devices, and that the two devices are unable to communicate during each round of the protocol.
Our analysis is based on the entropy accumulation theorem (EAT) \cite{DFR16} following closely its application to DIRNG in \cite{AFVR16}.

The EAT, as its name indicates, provides an estimate of the smooth min-entropy accumulated throughout a sequence of measurements.
It implies that the smooth min-entropy of the joint measurement outcomes of our protocol scales linearly with the number of rounds, with each round providing on average an amount of min-entropy roughly equivalent to the von Neumann entropy of a single round's outcome.
In order to use the EAT, it is first necessary to bound this single-round von Neumann entropy as a function of the expected probability of success $\omega$ in the tilted-CHSH game.
The following Lemma, which we derive in Appendix~\ref{app:robust} from the robust self-testing bounds for the tilted-CHSH inequality \cite{BP15}, provides a bound on the conditional min-entropy, which in turn bounds the conditional von Neumann entropy:

\begin{restatable}{lemma}{minentropylemma}
\label{lem:minentropy}
Let $\omega$ be the expected winning probability for the tilted-CHSH game with parameter $\beta$ of a pair of quantum devices, whose internal degrees of freedom can be entangled with the environment $E$.
Then the conditional min-entropy of the measurement outcome $\Oa$ for input $\Sa = 1$ is bounded as
\begin{equation}
 \Hmin(\Oa \mid E; \Sa=1) \geq 1 - \kappa \theta^{-4} \sqrt{\omega_\quant - \omega} \equiv g(\omega)
\end{equation}
with $\kappa \le 4 \sqrt{4+\beta} \, (4\sqrt2 + 61)/\ln 2 \leq 385 \sqrt{4+\beta}$.
\end{restatable}
The behavior of the bound with respect to $\omega_\quant - \omega$ is optimal \cite{RUV12}, while numerical results suggest that the optimal dependency in $\theta$ is $\bigo(\theta^{-2})$ \cite{CJA+17}.
This will not, however, significantly affect our conclusions.

Using this bound in the EAT along the lines of \cite{AFVR16} (see Appendix~\ref{app:EAT}) yields the following theorem:

\begin{restatable}{theorem}{minentropytheorem}
\label{thm:minentropy}
\twocolumnswitch{
	\def\etab{&}
	\def\ebreak{\nonumber\\}
	\let\eenv\align \let\endeenv\endalign
}{
	\def\etab{}
	\def\ebreak{}
	\let\eenv\equation \let\endeenv\endequation
}
Let $\vec \Oa, \vec \Ob, \vec \Sa, \vec \Sb, \vec T, \vec C$ be the classical random variables output by the protocol, and $E$ the quantum side information of a potential adversary.
Let $\success = \success(\vec C)$ be the success event for the protocol.
Let $\epsilon'$, $\epsilon_\smooth$ be two positive error parameters.
Then, for any given pair of devices used in the protocol, either $\Pr[\success] \le \epsilon'$ or
\begin{equation}\label{eq:hmin}
\Hmin^{\epsilon_\smooth}(\vec \Oa \vec \Ob\mid \vec \Sa \vec \Sb \vec T E; \success) \ge \nu \tau \kround \eqp,
\end{equation}
where $\nu = 1 - \gamma (2+\beta)/(4+\beta)$,
\begin{eenv}
	\tau = 1 \etab- \kappa \theta^{-4} \sqrt{\xi + \frac{2}{\gamma \sqrt \kround} \sqrt{1 - 2 \log_2(\epsilon_\smooth \epsilon')}}
\ebreak
		\etab- \frac{2 \log_2 26}{\sqrt \kround} \sqrt{1 - 2 \log_2(\epsilon_\smooth \epsilon')} \eqp,
\label{eq:rate}
\end{eenv}
and $\Hmin^{\epsilon_\smooth}(\vec \Oa \vec \Ob\mid \vec \Sa \vec \Sb \vec T E; \success)$ is the ${\epsilon_\smooth}$-smooth min-entropy of the output $(\vec \Oa, \vec \Ob)$ given $\vec \Sa, \vec \Sb,\vec T, E$ and conditioned on the event $\success$.
\end{restatable}

Given such a bound on the smooth min-entropy, there exist efficient procedures to extract from the raw outputs of the protocol a string of close-to-uniform random bits whose length is of the order of $\Hmin^{\epsilon_\smooth}$, with the smoothing parameter $\epsilon_\smooth$ characterizing the closeness to the uniform distribution.

\subsection{Random bits from any partially entangled two-qubit state}
Theorem~\ref{thm:minentropy} directly implies the following corollary, which shows the possibility of generating one bit of randomness per arbitrarily weakly entangled qubit pair:
\begin{corollary}
\label{cor:rng-partial}
For any constant values of the protocol parameters $\theta$, $\xi$, and $\gamma$ such that $\kappa \theta^{-4} \sqrt\xi < 1$ and for sufficiently large $\kround$, the protocol has vanishing completeness error and it generates $\bigomega(\kround)$ bits of randomness from $\kround$ partially entangled states $\ket{\psi_\theta}$.
For $\xi$ and $\gamma$ approaching $0$, the production of randomness in the protocol is asymptotically equal to $\kround$.
\qed
\end{corollary}

\subsection{Sublinear entanglement consumption}
We now consider how, in an ideal implementation of our protocol, the amount of shared entanglement consumed is related to the amount of randomness produced.
For given $\kround$, the entanglement consumption obviously decreases with smaller values of $\theta$.
According to \eqref{eq:hmin} and \eqref{eq:rate}, the randomness produced, however, also decreases with smaller $\theta$, unless this decrease is compensated by a suitable choice of the parameters $\gamma$ and $\xi$.
Indeed, for small $\theta$, $\gamma$ should be made larger to increase the fraction of game rounds and better test the devices.
Similarly, $\xi$ should be smaller (i.e., the threshold for the protocol's success should be set higher) in order for Lemma~\ref{lem:minentropy} to certify a nontrivial amount of min-entropy.
But the parameters $\gamma$ and $\xi$ also appear in the completeness error \eqref{eq:d-completeness} and thus cannot be set completely freely if this error is to remain small: setting $\gamma$ too low makes the estimation of the success rate at step \ref{item:protocol:estimation} of the protocol more uncertain,%
\footnote{%
From the perspective of randomness generation, a small value of $\gamma$ is desirable as it increases the factor $\nu$ in \eqref{eq:hmin} by increasing the rate of generation rounds.
However, game rounds also contribute to the final randomness, which makes the choice of $\gamma = 1$ possible.
}
and setting $\xi$ too low makes the threshold harder to reach.
In the following corollary to Theorem~\ref{thm:minentropy}, we show that there exists a choice for the parameters $\theta$, $\xi$, and $\gamma$, expressed as functions of $\kround$, such that the consumption of ebits $\kebit$ is sublinear in the number of rounds $\kround$:
\begin{corollary}
\label{cor:ebit}
Let $\lambda_\threshold$, $\lambda_\gamma$, $\lambda_\theta$ be positive scaling parameters such that
\begin{gather}
\label{eq:tradeoff-entanglement}
	\lambda_\theta < 2\lambda_\threshold \eqp. \\
\label{eq:tradeoff-completeness}
	\lambda_\threshold + \lambda_\gamma < 1/2 \eqp,
\end{gather}
Let $\theta = \kround^{-\lambda_\theta/16}$, $\xi = \kround^{-\lambda_\threshold}$, $\gamma = \kround^{-\lambda_\gamma}$, and constant $\epsilon_\smooth$ and $\epsilon'$.
Then, for $\kround \to \infty$, the entropy bound of Theorem~\ref{thm:minentropy} is asymptotically equal to $\kround$,
the completeness error vanishes,
and the amount of entanglement consumed is sublinear:
\begin{equation}
\label{eq:asymptotic-n-m}
	\kebit = \kround S(\theta) \sim \kround \theta^2\log_2\theta^{-2} =
	\frac{\lambda_\theta}{8} \kround^{k} \log_2 \kround \eqp,
\end{equation}
with $k = 1 - \lambda_\theta/8 \in \ooint{7/8, 1}$.
\qed
\end{corollary}

The constants $\lambda_\theta$, $\lambda_\xi$, $\lambda_\gamma$ give the rate at which the parameters $\theta$, $\xi$ and $\gamma$ tend to zero with increasing $\kround$.
The condition \eqref{eq:tradeoff-entanglement} expresses the fact that when the entanglement is small, the success threshold must be close to the maximum in order for the min-tradeoff function to take a nontrivial value (see Appendix~\ref{app:EAT}).
The condition \eqref{eq:tradeoff-completeness} expresses a tradeoff in the completeness error between how close the success threshold is to the maximum and the number of rounds that must be devoted to testing the correlations.
A larger fraction of game rounds (i.e., a larger $\gamma$) makes the success criterion fluctuate less, which allows for a higher threshold (i.e., a smaller $\xi$).

\section{Using diluted singlets}
\label{sec:dilution}
As mentioned in the introduction, the use of partially entangled states for randomness expansion enables us to reduce the amount of qubits exchanged between the devices when preparing their shared entanglement.
We reach this goal by applying our protocol to the outcome of an \emph{entanglement dilution} procedure, which transforms $\kebit$ singlet states $\ket{\phi^+}$ to $\kround \simeq \kebit / S(\theta)$ partially entangled states $\ket{\psi_\theta}$.
Thus, only $\kebit$ qubits need to be transferred between the devices in order to prepare the initial state $\ket{\phi^+}^{\otimes \kebit}$.

We use the procedure of Bennett et al.\ \cite{BBPS96}, in which Alice prepares the $\kround$ pairs locally, processes Bob's share with Schumacher compression then teleports them to Bob using the $\kebit$ singlets, who expands them back to $\kround$ qubits.
Since Schumacher compression is a lossy operation, the resulting state shared by Alice and Bob, which we denote as
$
	\mathcal D_{\theta,\delta}(\projop{\phi^+}^{\otimes \kebit})
$
is not exactly $\ket{\psi_\theta}^{\otimes \kround}$, but it is close in trace distance (with $\norm{\rho}_1 = \tr\abs\rho$):
\begin{restatable}{lemma}{dilutionlemma}
\label{lem:dilution}
\twocolumnswitch{
	\let\eenv\multline \let\endeenv\endmultline
}{
	\let\eenv\equation \let\endeenv\endequation
}
Using perfect devices, the dilution channel $\mathcal D_{\theta,\delta}$ maps $\kebit$ copies of the singlet $\ket{\phi^+}$ into $\kround$ copies of the partially entangled qubit state $\ket{\psi_\theta} = \cos\theta \ket{00} + \sin\theta \ket{11}$ with $\kebit = (S(\theta) + \delta) \kround$, up to error terms bounded by
\begin{eenv}
\label{eq:epsilonprep}
  \norm*{
    \mathcal D_{\theta,\delta}(\projop{\phi^+}^{\otimes \kebit})
      - \projop{\psi_\theta}^{\otimes \kround}
  }_1 \\
    \le 2 \sqrt{\epsilon_\proj} + \epsilon_\proj
    \equiv \epsilon_\prep
  \eqp,
\end{eenv}
with
\begin{gather}
\SwapAboveDisplaySkip
\label{eq:epsilon-proj}
  \epsilon_\proj = 2 \exp(-2 \kround \delta^2/\Delta^2) \eqp, \\
  \Delta = -\log_2 \tan^2 \theta \eqp.
\end{gather}
\end{restatable}
This lemma mostly follows from \cite{BBPS96,Wil13}; we prove it in Appendix~\ref{app:dilution}.

It follows from Lemma~\ref{lem:dilution} that even a perfectly implemented dilution procedure introduces some noise in the protocol.
We thus need to derive a new statement for the completeness error, the probability that perfect devices fail the protocol.
Using the indistinguishability interpretation of the trace distance, if the reference state $\ket{\psi_\theta}^{\otimes \kround}$ passes the threshold of the protocol with probability $1-\epsilon$, the diluted state $\mathcal D_{\theta,\delta}(\projop{\phi^+}^{\otimes \kebit})$, which is $\epsilon_\prep$-close to the reference, will pass the same threshold with probability at least $1-\epsilon - \tfrac12 \epsilon_\prep$ \cite{Wil13}.
Using the value of $\epsilon$ given in Lemma~\ref{lem:completeness} immediately implies the following:
\begin{lemma}
\label{lem:d-completeness}
Starting from $\kebit$ perfect singlets, the composition of the dilution procedure $\mathcal D_{\theta,\delta}$ and the randomness generation protocol has its completeness error bounded by
\begin{equation}
\label{eq:d-completeness}
	\epsilon_\compl = \tfrac12 \epsilon_\prep + \exp \mleft( -2 \kround (\gamma \xi)^2 \mright) \eqp.
\end{equation}
\end{lemma}

The following analogue of Corollary~\ref{cor:ebit} applies to the composition of entanglement dilution and randomness expansion; it immediately follows from the chosen parameterization:
\begin{corollary}
\label{cor:d-ebit}
Let $\delta = S(\theta) c$ with $c = \kround^{\lambda_c/8}$ for some real parameter $\lambda_c$.
Let the parameters of the protocol be set as in Corollary~\ref{cor:ebit}, with the additional constraint that
\begin{equation}
\label{eq:tradeoff-dilution}
	0 < \lambda_c < \lambda_\theta \eqp.
\end{equation}
Starting from $\kebit$ singlets, the composition of entanglement dilution with parameters $\delta$ and $\theta$ with the randomness expansion protocol yields an entropy bound in Theorem~\ref{thm:minentropy} which is asymptotically equal to $\kround$ for $\kround \to \infty$,
with a vanishing completeness error,
and a sublinear consumption of entanglement:
\begin{multline}
	\kebit
		=    \kround \, (S(\theta)+\delta) \\
		\sim \kround^{1+\lambda_c/8}\theta^2\log_2\theta^{-2}
		=     \frac{\lambda_\theta}{8} \kround^{k'} \log_2 \kround \eqp,
\end{multline}
with $k' = 1 - (\lambda_\theta-\lambda_c)/8 \in \ooint{7/8, 1}$.
\qed
\end{corollary}

In addition to the tradeoffs \eqref{eq:tradeoff-entanglement} and \eqref{eq:tradeoff-completeness}, which we discussed after Corollary~\ref{cor:ebit}, the bound $\lambda_c > 0$ ensures that the completeness error vanishes (which requires that $\epsilon_\proj$, defined in \eqref{eq:epsilon-proj}, also vanishes).
The upper bound $\lambda_c < \lambda_\theta$ ensures that the dilution process increases the number of states, i.e., $\kround > \kebit$.

\section{Robustness to noise}
While we have shown above that the inherent noise associated to dilution is tolerated by our protocol, we implicitly assumed that the quantum devices themselves are noise-free.
Indeed the completeness error given by Eq.~\eqref{eq:comperr} is evaluated assuming quantum devices with an expected winning probability equal to the quantum maximum $\omega_\quant$.
If the quantum devices are noisy, for example due to faulty measurements, and have instead a suboptimal winning probability $\omega=\omega_\quant-\zeta$ with $\zeta < \xi$, then the completeness error becomes
\begin{equation}\label{eq:comperr2}
	\epsilon_\compl = \exp \mleft( -2 \kround \gamma^2 (\xi-\zeta)^2 \mright) \eqp.
\end{equation}
It is easy to see that a sublinear entanglement consumption remains possible with such noisy devices, provided the noise parameter $\zeta$ decreases with $\kround$ as $\zeta=\kround^{-\lambda_\zeta}$ for some suitable scaling parameter $\lambda_\zeta$.
However, realistic devices will be subject to a constant amount of noise, rather than an asymptotically vanishing one.
In this case, the protocol as described so far breaks down.
This is most easily seen from Lemma~\ref{lem:minentropy}: it is clear that $\theta$ can only be taken as low as values of the order of $(\omega_\quant-\omega)^{1/8}=\zeta^{1/8}$ to get a nontrivial bound on the min-entropy.
Nevertheless, given a small enough finite upper bound on the amount of noise $\zeta$, partially entangled qubit pairs with an appropriate value of $\theta < \pi/4$ can still be used to produce a linear amount of randomness with a yield per ebit higher than $1$ according to Theorem~\ref{thm:minentropy}, thus improving what can directly be achieved using maximally entangled states.

A sublinear consumption of entanglement using diluted singlets can be recovered even with devices whose components fail with constant probability if our protocol is combined with error correction and fault-tolerant quantum computation.
Indeed, according to the threshold theorem for fault-tolerant quantum computation, an arbitrary quantum circuit containing $G(n)$ gates may be simulated with probability of error $e(n)$ on hardware whose components fail with constant probability at most $p$, provided $p$ is below some threshold, through an encoding that increases the local dimension of each qubit by a factor $\operatorname{poly} \log G(n)/e(n)$ \cite{NC00}.
In our case, the number $G(\kround)$ of gates needed to perform entanglement dilution \cite{CD96} and the subsequent bipartite measurements is polynomial in $\kround$.
On the other hand, aiming for a probability of error $e(\kround)$ that decreases polynomially in $\kround$ for the simulating circuit yields a completeness error that vanishes asymptotically.
The number of ebits needed in such a fault-tolerant version of our protocol is then multiplied only by a factor $\operatorname{poly} \log G(\kround)/e(\kround)=\operatorname{poly} \log \kround$ resulting in a total number of ebits that is still sublinear in $\kround$, i.e., $\kebit\sim \kround^k \operatorname{poly} \log \kround$.

\section*{Discussion}
In summary, earlier work \cite{AMP12} showed that a pair of entangled qubits with arbitrarily little entanglement could be used to certify one bit of randomness.
Here we have carried out the further step of transforming the intuition of \cite{AMP12} into DIRNG protocols in which a sublinear amount of entanglement is consumed.
This shows that the consumption of entanglement resources in DIRNG can be dramatically improved qualitatively and quantitatively with respect to existing protocols.
These results about entanglement are analogous to those concerning the initial random seed, the other fundamental resource required in DIRNG.
Interestingly, the recent work \cite{CJA+17} suggests that one could devise DIRNG protocols that consume a constant amount of entanglement.
Whether the intuition of \cite{CJA+17} can be transformed into such a DIRNG protocol is an interesting open question.

In the present work, we did not attempt to minimize simultaneously the entanglement and the size of the initial seed.
(Note that in our protocol, the size of the initial random seed is determined by the parameter $\gamma$ specifying the proportion of test rounds.)
Nevertheless, in the parameter regimes of Corollaries~\ref{cor:ebit} and \ref{cor:d-ebit}, the entanglement and the initial seed are both sublinear.
Interestingly, it appears that our approach involves a tradeoff between entanglement and seed consumption, given the constraints placed on $\lambda_\gamma$ and $\lambda_\theta$ in Corollaries~\ref{cor:ebit} and \ref{cor:d-ebit}.
Indeed, equations \eqref{eq:tradeoff-entanglement} and \eqref{eq:tradeoff-completeness} imply that $\lambda_\theta+2 \lambda_\gamma<1$.
Thus if $\lambda_\theta$ is close to $1$ (corresponding to a small consumption of entanglement), $\lambda_\gamma$ must be close to $0$, which indicates a high proportion $\gamma$ of test rounds and, as a result, high consumption of random seed.
Likewise, if $\lambda_\gamma$ is close to $1/2$, $\lambda_\theta$ must be close to $0$, and the protocol requires high entanglement and low random seed.
We leave as open questions whether there is some kind of fundamental tradeoff between the required amounts of random seed and shared quantum resources, and whether the amount of quantum resources in our protocol is optimal or can be further decreased.

\paragraph{Acknowledgments.}
This work is supported by the Fondation Wiener-Anspach, the Interuniversity Attraction Poles program of the Belgian Science Policy Office under the grant IAP P7-35 photonics@be.
S. P. is a Research Associate of the Fonds de la Recherche Scientifique (F.R.S.-FNRS).
C. B. acknowledges funding from the F.R.S.-FNRS through a Research Fellowship.

\printbibliography%

\onecolumn\newpage
\appendix
\input{appendix_text.tex}

\end{document}

%% file: h_macros.tex
\newtheorem{theorem}{Theorem}
\newtheorem{corollary}[theorem]{Corollary}
\newtheorem{lemma}[theorem]{Lemma}

\renewcommand{\vec}{\mathbf}

\DeclarePairedDelimiter{\abs}{\lvert}{\rvert}
\DeclarePairedDelimiter{\norm}{\lVert}{\rVert}
\DeclarePairedDelimiter{\ceil}{\lceil}{\rceil}
\DeclarePairedDelimiter{\mean}{\langle}{\rangle}

\let\lsbkt[ \let\rsbkt]
\let\lparn( \let\rparn)
\DeclarePairedDelimiter{\coint}{\lsbkt}{\rparn}

\DeclarePairedDelimiter{\ccint}{\lsbkt}{\rsbkt}
\DeclarePairedDelimiter{\ooint}{\lparn}{\rparn}

\newcommand{\ketbra}[2]{\ket{#1}\!\bra{#2}}
\newcommand{\projop}[1]{\ketbra{#1}{#1}}

\newcommand{\eqp}{\,}

\newcommand{\rmin}{\mathrm{min}}

\newcommand{\rma}{\mathrm{a}}

\newcommand{\prep}{\mathrm{prep}}
\newcommand{\proj}{\pi}
\newcommand{\quant}{\mathrm{q}}
\newcommand{\compl}{\mathrm{c}}
\newcommand{\smooth}{\mathrm{s}}
\newcommand{\tangent}{\mathrm{t}}
\newcommand{\threshold}{\xi}

\newcommand{\junk}{\text{junk}}

\DeclareMathOperator{\tr}{Tr}

\newcommand{\id}{I}

\DeclareMathOperator{\idmap}{Id}

\newcommand{\oa}{{a}}
\newcommand{\ob}{{b}}
\newcommand{\sa}{{x}}
\newcommand{\sbb}{{y}}
\newcommand{\Oa}{{A}}
\newcommand{\Ob}{{B}}
\newcommand{\Sa}{{X}}
\newcommand{\Sb}{{Y}}

\newcommand{\kround}{n}
\newcommand{\kebit}{m}

\newcommand{\sA}{{\mathrm A}}
\newcommand{\sB}{{\mathrm B}}
\newcommand{\sE}{{\mathrm E}}

\newcommand{\success}{\mathcal{S}}
\newcommand{\Hmin}{H_\rmin}

\newcommand{\bigo}{O}
\newcommand{\bigomega}{\Omega}

\newcommand{\pguess}{p_{\text{guess}}}

\newcommand{\twocolumnswitch}{\ifdefstring{\thepagegrid}{mlt}}

%% file: appendix_text.tex
\section{Proof of Lemma~\ref{lem:minentropy}}\label{app:robust}
Lemma~\ref{lem:minentropy} gives a lower bound on the conditional min-entropy in the outcome of Alice's measurement $\sa = 1$ for devices which achieve a certain success probability $\omega$ in the tilted-CHSH game of parameter $\beta$.
We restate it here:
\minentropylemma*
\begin{proof}
To derive this bound, we use our self-testing result for the tilted-CHSH inequalities \cite{BP15}.
The robustness bounds for the self-test in \cite{BP15} are rather unwieldy, so we will instead provide a crude upper bound that retains the same asymptotic behavior in $\beta$ and $\omega$ and greatly simplifies the use of the bound.

We will lower-bound $\Hmin(\Oa \mid E ; \Sa = 1)$ as a function of the expected violation of the tilted-CHSH inequality, $I = I_\quant - \epsilon$, where $I_\quant = \sqrt{8+2\beta^2} = 4/\sqrt{1+\sin^2 2\theta}$ is the maximal quantum value of the expression.
This min-entropy is equivalent to the guessing probability for the measurement $\Sa=1$, which is defined as
\begin{equation}
\label{eq:gp-minentropy}
2^{-\Hmin(\Oa \mid E ; \Sa = 1)}
	= \pguess(\Oa \mid E ; \Sa = 1)
	= \max_{\{M_g\}} \Pr[\Oa = G \mid \Sa = 1] \eqp,
\end{equation}
where $\{M_g\}$ is a POVM on the subsystem $\sE$, which an adversary would use to measure the side information contained in $\sE$ to formulate a guess $G$ for $\Oa$ \cite{KRS09}.
Formulated in terms of a given physical state $\ket{\tilde\psi}_{\sA\sB\sE}$ and observables $\tilde A_{\sa} \equiv \tilde A_{\sa} \otimes \id_\sB \otimes \id_\sE$ and $\tilde B_{\sbb} \equiv \id_\sA \otimes \tilde B_{\sbb} \otimes \id_\sE$, for a given adversary POVM $\{M_g\}$ we have
\begin{align}
\Pr[\Oa = G \mid \Sa = 1]
	&= \sum_{\oa \in \{0,1\}} \bra{\tilde\psi} \frac{\id_\sA + (-1)^{\oa} \tilde A_1}{2} \otimes \id_\sB \otimes M_\oa \ket{\tilde\psi} \\
	&= \tfrac12 + \tfrac12 \bra{\tilde\psi} \tilde A_1 \otimes \id_\sB \otimes (M_0-M_1) \ket{\tilde\psi}
	\equiv \tfrac12 + \tfrac12 \mean{\tilde A_1 C} \eqp,
\label{eq:pguess}
\end{align}
letting $C = M_0-M_1$, which is bounded as $\norm{C}_\infty \le 1$.
(From here on we will sometimes use a shorter notation where instead of e.g.\ $\id_\sA \otimes \id_\sB \otimes C$, we simply write $C$.)
We will now relate this expression to the reference system using self-testing.

Using the notation of \cite{BP15}, the self-testing result shows that any pure state $\ket{\tilde\psi}_{\sA\sB\sE}$ measured by $\tilde A_\sa$ and $\tilde B_\sbb$ in such a way that the tilted-CHSH inequality for these observables is violated up to $I_\quant - \epsilon$ obeys%
\footnote{In \cite{BP15} the subsystem $\sE$ is implicitly included in $\sA$ and/or $\sB$ as a purifying subsystem for the mixed state held and measured by the devices.
The statement can easily be modified to separate it without changing the proofs; the black-box measurement operators $\tilde A_\sa$ and $\tilde B_\sbb$ then act as the identity on $\sE$.}
\begin{subequations}
\label{eq:selftest}
\begin{align}
\norm{
	\Phi(\ket{\tilde\psi}_{\sA\sB\sE})
	- \ket{\psi_\theta}_{\sA'\sB'} \otimes \ket{\junk}_{\sA\sB\sE}
}
	&\le 2 \bar\delta
\eqp,\\
\norm{
	\Phi(\tilde A_1 \ket{\tilde\psi}_{\sA\sB\sE})
	- A_1 \ket{\psi_\theta}_{\sA'\sB'} \otimes \ket{\junk}_{\sA\sB\sE}
}
	&\le 2 \bar\delta + 2 \delta_\rma^\sA
\eqp,
\end{align}
\end{subequations}
where $\Phi = \Phi_\sA \otimes \Phi_\sB \otimes \id_\sE$ is a local isometry acting on Alice and Bob's subsystems which introduces and transforms ancillary qubits $\sA'$ and $\sB'$,
$\ket{\psi_\theta}$ is the reference state (see main text),
$A_1$ is the reference observable that yields one bit of randomness,
and the error bound parameters $\bar\delta, \delta_\rma^\sA = \bigo(\sqrt\epsilon \theta^{-4})$ are explicitly defined in \cite{BP15}.
Effectively, this isometric transformation extracts a state onto $\sA'\sB'$ which is close to a copy of the reference state and almost decorrelated from the initial system $\sA\sB\sE$.
Likewise, the isometry approximately maps the physical observables' action on the physical state in $\sA\sB$ to ideal actions on the reference state in the ancillary registers $\sA'\sB'$.

From this result, we see that the guessing probability with respect to $\ket{\tilde\psi}$ (which is by isometry identical to the guessing probability for $\Phi(\ket{\tilde\psi})$) is close to the guessing probability with respect to the reference state, for which $\pguess = \max_\oa \Pr[\Oa = \oa \mid \Sa = 1] = 1/2$ since the side information $E$ is decorrelated from the devices $\sA$ and $\sB$.

To show this approximate equality of guessing probabilities, we rewrite the last term of \eqref{eq:pguess} as
\begin{equation}
	\mean{\tilde A_1 C}
	= \Phi^\dag(\bra{\tilde\psi}) \, C \, \Phi(\tilde A_1 \ket{\tilde\psi})
\end{equation}
using that $\Phi$ is an isometry which acts like the identity on $\sE$.
We then use \eqref{eq:selftest} and the triangle inequality to transform this into $\bra{\psi_\theta} A_1 \ket{\psi_\theta} \braket{\junk | C | \junk}$, with additional error terms.
First, we replace $\Phi(\tilde A_1 \ket{\tilde\psi})$ with $A_1 \ket{\psi_\theta} \otimes \ket{\junk}$ with an additional error term $2\bar\delta + 2\delta_\rma^\sA$ since $\norm{\Phi^\dag(\bra{\tilde\psi}) \, C} \le 1$, then we replace $\Phi^\dag(\bra{\tilde\psi})$, again using $\norm{A_1}_\infty, \norm{C}_\infty \le 1$:
\begin{align}
\mean{\tilde A_1 C}
	&\le \Phi^\dag(\bra{\tilde\psi}) \, \bigl( A_1 \ket{\psi_\theta} \otimes C \ket{\junk} \bigr) + 2\bar\delta + 2\delta_\rma^\sA \\
	&\le \bra{\psi_\theta} A_1 \ket{\psi_\theta} \bra{\junk} C \ket{\junk} + 4\bar\delta + 2\delta_\rma^\sA \\
	&\le \abs{\bra{\psi_\theta} A_1 \ket{\psi_\theta}} + 4\bar\delta + 2\delta_\rma^\sA \eqp.
\end{align}
Since for the reference system $\bra{\psi_\theta} A_1 \ket{\psi_\theta} = 0$, we find that
\begin{equation}
\label{eq:pguess-delta}
	\pguess \le \tfrac12 + 2\bar\delta + \delta_\rma^\sA
\end{equation}

We now proceed to find a simple expression of $\epsilon$ and $\theta$ that upper-bounds $2\bar\delta + \delta_\rma^\sA$.
After some careful manipulation of the rather long expressions for $\bar\delta$ and $\delta_\rma^\sA$, we find
{
\newcommand{\sqs}{\sqrt{1+s^2}}
\begin{multline}
\label{eq:longbound}
  2\bar\delta + \delta_\rma^\sA
    = \sqrt{2 I_\quant} \sqrt\epsilon \Biggl[
      \frac{\sqs}{2s^2} \left( 1 + c + \sqs \right)
      + \frac{\sqs}{4s} \left( 2 - c + \sqs \right)
    \\
      + \frac{c + \sqs}{2s^2} (1+c)
        \left( 8 + 2 \frac{1+\sqs}{s^2} + 3 \tan\theta \right)
    \Biggr]
  \eqp,
\end{multline}
}%
with $c = \cos(2 \theta)$, $s = \sin(2 \theta)$, $I_\quant = 4 / \sqrt{1+s^2}$.
The dominating term in this bound for small $\theta$ comes from the term in $s^{-4}$, namely $
  2 \sqrt{2 I_\quant} \sqrt\epsilon (1+c)(c+\sqrt{1+s^2})(1+\sqrt{1+s^2}) s^{-4} = \bigo(\sqrt\epsilon \theta^{-4})
$.

A crude upper bound on \eqref{eq:longbound} is obtained by taking a $s^{-4}$ factor out of the square brackets and giving rough numerical bounds on the bounded function that remains.
For instance, the factor of $s^{-4} \sqrt{2 I_\quant} \sqrt\epsilon$ in the first term becomes $
  s^2 \sqrt{1+s^2} (1+c+\sqrt{1+s^2})/2 \le (3\sqrt2/2) s^2
$ because $
  c + \sqrt{1+s^2} = \sqrt{1-s^2} + \sqrt{1+s^2} \le 2
$ for $s^2 \in [0,1]$, and $\sqrt{1+s^2} \le \sqrt{2}$.
We obtain the following bound:
\begin{equation}
  2\bar\delta + \delta_\rma^\sA
    \le 2\sqrt2 \mleft[
      \frac{3\sqrt2}{2} s^2
      + \frac{1+\sqrt2}{2} s^3
      + 16 s^2 + 8 + 6 s^2 \tan\theta
    \mright]
    \frac{\sqrt\epsilon}{s^4}
  \eqp,
\end{equation}
where we have also bounded $I_\quant \le 4$ and used the following tight bounds after expanding the third term in the square brackets of \eqref{eq:longbound}:
\begin{gather}
  (c + \sqrt{1+s^2})(1+c) \le 4 \eqp,
  \\
    (c + \sqrt{1+s^2})(1+c)(1+\sqrt{1+s^2})
      = (1+c\sqrt{1+s^2})(2+c+\sqrt{1+s^2})
      \le 8 \eqp.
\end{gather}
The factor $\tan\theta$ in the last term is simply bounded by $1$.
The bound we reach is the following:
\begin{equation}
  2\bar\delta + \delta_\rma^\sA
    \le 2\sqrt2 \mleft[
      \frac{1+\sqrt2}{2} s^3
      + \frac{3\sqrt2 + 44}{2} s^2
      + 8
    \mright] \frac{\sqrt\epsilon}{s^4}
  \eqp.
\end{equation}
Finally, the polynomial in $s$ in square brackets is bounded by its maximum at $s=1$.
Eq.~\eqref{eq:pguess-delta} then becomes
\begin{equation}
\label{eq:shortbound}
  \pguess - \tfrac12
    \le 2\bar\delta + \delta_\rma^\sA
    \le
    (8 + 61 \sqrt2) \frac{\sqrt\epsilon}{s^4}
  \eqp.
\end{equation}
Further bounding
\begin{equation}
	s = \sin 2\theta \ge \frac{2\theta}{\pi/2} \ge \theta
\end{equation}
by concavity of the sine function on $\ccint{0,\pi/2}$ and substituting $\epsilon = (8+2\beta) (\omega_\quant - \omega)$, we reach our final bound for the guessing probability,
\begin{equation}
\pguess \le \tfrac12 + \sqrt{8+2\beta} \, (8+61\sqrt2) \, \theta^{-4} \sqrt{\omega_\quant - \omega} \eqp.
\end{equation}
Putting this together with \eqref{eq:gp-minentropy}, we find, using $\ln(1+x) \le x$,
\begin{align}
\Hmin(\Oa \mid E ; \Sa = 1)
	&\ge 1 - \log_2 \mleft( 1 + 2 \sqrt{8+2\beta} \, (8+61\sqrt2) \, \theta^{-4} \sqrt{\omega_\quant - \omega} \mright) \\
	&\ge 1 - \kappa \theta^{-4} \sqrt{\omega_\quant - \omega}
\end{align}
with $\kappa = 4 \sqrt{4+\beta} (4\sqrt2+61) / \ln 2$.
\end{proof}

A numerical maximization of the factor of $\sqrt\epsilon$ in \eqref{eq:longbound} shows that a tighter numerical factor of $45.13$ could replace the numerical factor $8+61\sqrt2 = 94.27$ in \eqref{eq:shortbound}, or less if the range of $\theta$ is limited.

\section{Proof of Theorem~\ref{thm:minentropy}}\label{app:EAT}
It is shown in \cite{DFR16,AFVR16} that obtaining a bound on the smooth min-entropy produced by a generic protocol of the type that we consider here reduces to finding a \emph{min-tradeoff function}, a certain function that bounds the randomness produced in an average round of the protocol.
This function is specific to the particular game used in the protocol and obtaining it for the tilted-CHSH game is the only part of the general analysis of \cite{AFVR16} that we need to tailor to our situation.

The min-tradeoff function is defined as follows.
Any protocol round (i.e., step~\ref{item:protocol:round} in the protocol; see Section~\ref{sec:protocol}) can be thought of as a quantum channel $\mathcal{G}_i$ mapping the state $\rho=\rho_{DE}$ of the pair of devices $D$ and the adversary information $E$ before that round to the resulting state $\mathcal{G}_i(\rho)=\rho'=\rho'_{\Oa_i\Ob_i\Sa_i\Sb_i T_iC_i D'E}$ after the protocol round, which also includes explicitly the classical data that was produced in that round.
In particular, the channel $\mathcal{G}_i$ and the initial state $\rho$ determine the probability distribution $p \equiv (p_0, p_1, p_\bot)$ for the classical random variable $C_i\in\{0,1,\perp\}$.
This probability distribution is related to the randomness produced in the protocol round: from Lemma~\ref{lem:minentropy} we expect that a pair of devices which succeeds at any round ($C_i = 1$) with higher probability produces more entropy in its outputs.
The min-tradeoff function is a function $f_\rmin(p)$ that bounds the randomness produced in the protocol round solely on the basis of the probability distribution $p$.
Formally, a function $f_\rmin(p)$ is a min-tradeoff function if it satisfies
\begin{equation}
	f_\rmin(p) \le H(\Oa_i\Ob_i \mid \Sa_i\Sb_i T_i E)_{\mathcal{G}_i(\rho)} \eqp,
\end{equation}
where $H(\Oa_i\Ob_i \mid \Sa_i\Sb_i T_i E)_{\mathcal{G}_i(\rho)}$ is the von Neumann entropy of the joint outputs conditioned on the classical side information produced in the round and on the quantum information of the adversary $E$.
This inequality should hold for all channels $\mathcal{G}_i$ that are compatible with the protocol and for all initial states $\rho$ such that the variable $C_i$ in $\mathcal{G}_i(\rho)$ is distributed as $p$.

Theorem~\ref{thm:minentropy}, which we restate here, follows from the entropy accumulation theorem \cite{DFR16} and its application to randomness generation protocols by Arnon-Friedman et al.~\cite{AFVR16}.
Our proof follows that of \cite{AFVR16}, in which we substitute a min-tradeoff function adapted to our protocol.
\minentropytheorem*
\begin{proof}[Proof of Theorem~\ref{thm:minentropy}]
We first note that $C_i = \bot$ happens if and only if the protocol round is a generation round, hence we always have $p_\bot = 1-\gamma$ and $p_0 + p_1 = \gamma$.
Thus, as noted in \cite{AFVR16}, we are free to define $f_\rmin(p)$ to arbitrary values when $p_0 + p_1 \ne \gamma$, since such a distribution for $C_i$ is not compatible with our protocol anyway.
On the other hand, when $p_0 + p_1 = \gamma$, the expected probability of succeeding at the nonlocal game in a game round is $p_1 / \gamma$.
In that case, we can use Lemma~\ref{lem:minentropy} with $\omega = p_1/\gamma$ to set the value of $f_\rmin$.
Indeed,
\begin{align}
	H(\Oa_i\Ob_i \mid \Sa_i\Sb_i T_i E)
	&\geq H(\Oa_i \mid \Sa_i \Sb_i T_i E) \label{eq:vn-trace} \\
	&=    H(\Oa_i \mid \Sa_i E) \label{eq:vn-flag} \\
	&\ge \Pr[\Sa_i = 1] \, H(\Oa_i \mid E;\Sa_i = 1) \\
	&\ge \nu \, g(\omega) \eqp,\label{eq:bound}
\end{align}
with $\nu = \Pr[\Sa_i = 1] = 1 - \gamma (2+\beta)/(4+\beta)$.
In \eqref{eq:vn-trace}, we used the chain rule and the positivity of the conditional entropy of classical information.
In \eqref{eq:vn-flag}, we used that Alice's output $\Oa_i$ is independent of Bob's measurement choice $\Sb_i$ and of the round flag $T_i$.
To get \eqref{eq:bound}, we used Lemma~\ref{lem:minentropy} and the fact that the conditional min-entropy lower-bounds the conditional von Neumann entropy \cite[Proposition~4.3]{Tom12}.

We can thus define $f_\rmin(p)=\nu\,g(p_1/\gamma$) when $p_0+p_1=\gamma$.
For convenience, we set it to the same value when $p_0+p_1\neq\gamma$, since it can be freely chosen in that case.
All in all,
\begin{equation}
\label{eq:cvx-mintr}
	f_\rmin(p) = \nu\,g(p_1/\gamma) \eqp.
\end{equation}

The EAT \cite{DFR16} requires \emph{affine} min-tradeoff functions.
Since $g$ is convex, we can simply obtain affine lower bounds of \eqref{eq:cvx-mintr} by taking its tangent at any point.
The tangent of $g(\omega)$ at the point $\omega = \omega_\tangent$ is given by
\begin{equation}
	\bar g_{\omega_\tangent}(\omega) = 1 - \kappa \theta^{-4} \frac{2\omega_\quant - \omega_\tangent - \omega}{2\sqrt{\omega_\quant - \omega_\tangent}} \eqp,
\end{equation}
hence the min-tradeoff function we finally use will be $f_\rmin(p) = \nu\,\bar g_{\omega_\tangent}(p_1/\gamma)$ for some appropriately chosen $\omega_\tangent$.

Given such a min-tradeoff function, Lemma~9 of \cite{AFVR16} then states that for any given pair of devices, either the protocol succeeds with low probability $\Pr[\success] \le \epsilon'$, or
\begin{equation}
	\Hmin^{\epsilon_\smooth}(\vec\Oa \vec\Ob \mid \vec\Sa \vec\Sb \vec T E ; \success) \ge \kround \nu\,\bar g_{\omega_\tangent}(\omega_\quant-\xi) - \mu \sqrt \kround
\end{equation}
with
\begin{equation}
	\mu = 2 \mleft( \log_2(13) + \ceil{\norm{\nabla f_\rmin}_\infty} \mright) \sqrt{1-2 \log_2(\epsilon_\smooth \epsilon')} \eqp.
\end{equation}
The gradient of the min-tradeoff function is simply the slope of $\nu\,\bar g_{\omega_\tangent}(p_1/\gamma)$:
\begin{equation}
	\norm{\nabla f_\rmin}_\infty = \gamma^{-1} \nu \frac{\kappa \theta^{-4}}{2 \sqrt{\omega_\quant - \omega_\tangent}} \eqp.
\end{equation}
Bounding the ceiling function as $\ceil x \le x + 1$ and optimizing over the point of tangency $\omega_\tangent$ produces the final expression \eqref{eq:hmin} for the min-entropy bound of the theorem.
\end{proof}

\section{Proof of Lemma~\ref{lem:dilution}}\label{app:dilution}

In Section~\ref{sec:dilution}, we sketched the entanglement dilution procedure of Bennett et al.~\cite{BBPS96}, which defines a channel $\mathcal D_{\theta,\delta}$ that approximately dilutes $\ket{\phi^+}^{\otimes \kebit}$ into $\ket{\psi_\theta}^{\otimes n}$, with $\kebit < \kround \simeq \kebit/S(\theta)$.

In this appendix, we describe this procedure in detail then prove Lemma~\ref{lem:dilution}, which bounds its inherent error terms.
We restate the Lemma here:
\dilutionlemma*

As stated in the main text, dilution is enabled by the possibility of compressing a number of weakly entangled states into a smaller Hilbert space at the cost of a small error.
One procedure that realizes this is known as Schumacher compression \cite{Sch95} and is explained in great detail in \cite{Wil13}, from which we borrow the notation.

We will apply Schumacher compression to the second half of the global state of $\kround$ weakly entangled states $
  \ket{\psi_\theta}^{\otimes \kround}_{\sA\sA'}
    = ( \cos\theta \ket{00} + \sin\theta \ket{11} )^{\otimes \kround}
$.
This allows us to use $\kebit < \kround$ maximally entangled qubit pairs to transport that second half from box $\sA$ to box $\sB$ so that the initial $\kebit$ singlets are effectively transformed into $\kround$ weakly entangled pairs after decompression.

Schumacher compression works on a source of pure states $\{\ket{\psi_i},q_i\}$ which outputs the states in $\{\ket{\psi_i}\}$ at random, with respective probabilities $\{q_i\}$.
The goal is to pack the information output by $\kround$ i.i.d.\ uses of the source into a smaller Hilbert space in a way that makes it recoverable later with high fidelity---%
a generalization of Shannon's source coding to the quantum setting.
We interpret the reduced density operator $\sigma_{\sA'} = \tr_{\sA} \projop{\psi_\theta}_{\sA\sA'}$ for the second half of one pair $\ket{\psi_\theta}_{\sA\sA'}$ as describing a quantum source of mutually orthogonal states---namely, the eigenstates of $\sigma$---with probabilities given by the corresponding eigenvalues.
The eigenstates of $\sigma$ coincide with the computational basis: they are $\ket 0$ and $\ket 1$, with corresponding eigenvalues $q_0 = \cos^2 \theta$ and $q_1 = \sin^2 \theta = 1-q_0$.
Using this source $\kround$ times gives us the mixed state $\rho_{\sA'} = \sigma_{\sA'}^{\otimes \kround} = \tr_{\sA} \projop{\psi_\theta}^{\otimes \kround}$, i.e., the reduced state of Bob's share of the set of $\kround$ entangled pairs prepared by Alice.
The eigenstates of $\rho_{\sA'}$ are the computational basis states for $\kround$ qubits, which we write as $\ket{y}$  with $y \in \{0, 1\}^\kround$, where $\ket{y}$ is the tensor product of $\kround$ qubits $\ket{y_1} \dotsm \ket{y_\kround}$.
The eigenvalue associated with $y$ is $\lambda_j = (\cos^2\theta)^j (\sin^2\theta)^{\kround-j}$ for $j = n(0\mid y)$, which gives the number of zeros in the binary string $y$.

Compressing a source of orthogonal states is more or less equivalent to Shannon source coding.
The idea is to consider the string $y$ obtained after $\kround$ uses of the source and only let it through when it is deemed ``typical'' enough.
According to the theory of typical sequences \cite{CT12}, there is a family of subsets of the $2^\kround$ strings $y$, called the $\delta$-typical subsets, that each contain an exponentially small fraction of strings but nevertheless have an exponentially large probability weight.
Thus, while most strings we obtain are typical, their number is considerably smaller than $2^\kround$.
The $\delta$-typical subset is defined as
\begin{equation}
\label{eq:typical}
  \mathcal T_\delta
    = \left\{
        y \in \{0,1\}^\kround :
          S - \delta \le \frac{-\log_2 P(y)}{\kround} \le S + \delta
      \right\}
  \eqp,
\end{equation}
where $S$ is the entropy of the source (which is also the entropy of entanglement of a single pair $\ket{\psi_\theta}$),
\begin{align}
  S = h_2(q_0)
    &= -q_0 \log_2(q_0) - q_1 \log_2(q_1) \\
    &= - \log_2(q_0) + q_1 \Delta \eqp,
\end{align}
with $\Delta = \log_2(q_0/q_1) = -\log_2 \tan^2\theta$.
Thus, a sequence $y$ is $\delta$-typical if and only if its sample entropy $-(1/\kround) \log_2 P(y)$ is $\delta$-close to $S$.
Since our ideal source is i.i.d., each random variable $Y_i$ is distributed independently according to the same Bernoulli distribution of probabilities $\{q_0, q_1\}$.
Thus, the sample entropy for a given value $y$ can be rewritten as
\begin{align}
-\frac1\kround \log_2 P(y)
  &=
  \frac1\kround \sum_{i=1}^\kround -\log_2 q_{y_i} \\
    &= -\frac{n(0\mid y)}{\kround} \log_2 q_0
      - \frac{n(1\mid y)}{\kround} \log_2 q_1
  \\
    &= -\log_2 q_0 + \frac{n(1\mid y)}{\kround} \Delta \eqp.
\end{align}
Hence, the definition of the typical set \eqref{eq:typical} can be rewritten as
\begin{equation}
\label{eq:typicalvar}
  \mathcal T_\delta
    = \left\{
        y \in \{0,1\}^\kround :
          q_1 - \frac\delta\Delta
            \le \frac{1}{\kround} \sum_{i=1}^\kround y_i \le q_1 + \frac\delta\Delta
      \right\}
  \eqp.
\end{equation}
That is, a sequence $y$ is $\delta$-typical if and only if it has a frequency of $1$'s that is $(\delta/\Delta)$-close to the expected value $q_1$.

Properties of the typical set are easily derived from those two expressions of $\mathcal T_\delta$ \cite{CT12}.
First, applying Hoeffding's inequality \cite{Hoe63} for the binomially-distributed $\sum_i Y_i$ shows that the typical set has a high probability weight
\begin{equation}
\label{eq:typicalprob}
  \Pr[\mathcal T_\delta]
    \ge 1 - \epsilon_\proj
  \eqp,
\end{equation}
where we defined the \emph{projection error}
\begin{equation}
  \epsilon_\proj = 2 \exp(-2 \kround \delta^2/\Delta^2) \eqp,
\end{equation}
which is an upper bound on the probability of a sequence being atypical.

Secondly, in contrast to this first property, the typical set has a relatively low cardinality: the number of typical sequences is exponentially small compared to the total number of sequences.
Indeed, from the definition \eqref{eq:typical},
\begin{align}
  \abs{\mathcal T_\delta}
    &=  2^{\kround(S+\delta)} \sum_{y\in\mathcal T_\delta} 2^{-\kround(S+\delta)}
  \\&
    \le 2^{\kround(S+\delta)} \sum_{y\in\mathcal T_\delta} P(y)
    \le 2^{\kround(S+\delta)} \eqp,
\label{eq:typicalcard}
\end{align}
which is much smaller than the total number of $2^\kround$ sequences if $S < 1-\delta$ and $\kround$ is high.

Source coding consists in discarding atypical sequences, which occur with low probability, and encoding typical sequences into smaller codewords.
This encoding is possible because of the small cardinality of the typical set: a sequence can simply be encoded by its index within a given ordering of the elements of the typical set, which gives binary codewords a length of at most $\log_2 \abs{\mathcal T_\delta} \le \kround(S+\delta) \equiv \kebit$.

Schumacher compression applies this procedure to the quantum state $
  \rho_{\sA'}
  = \sum_{y \in \{0,1\}^\kround} \lambda_{n(0 \mid y)} \projop{y}
$, which describes the output of a quantum source of pure states $\ket y$ with probabilities $\lambda_{n(0 \mid y)}$.
In order to identify an atypical state, a projective \emph{typicality measurement} is performed, with projectors $\{\Pi_\delta, \id - \Pi_\delta\}$ where
\begin{equation}
  \Pi_\delta = \sum_{y\in\mathcal T_\delta} \projop{y} \eqp.
\end{equation}
If the typicality measurement succeeds, the state ends up in the typical subspace spanned by $\{ \ket{y} : y \in \mathcal T_\delta\}$, and can be encoded in a $2^\kebit$-dimensional Hilbert space by an invertible isometric map $V$.
If instead the measurement fails, a given typical state $\tau$ is substituted and encoded the same way.
The resulting state is therefore $
  \mathcal C(\rho_{\sA'})
    = V [ \Pi_\delta \rho_{\sA'} \Pi_\delta
      + (1-\tr(\Pi_\delta \rho_{\sA'})) \, \tau ] V^\dag
$.

The two properties \eqref{eq:typicalprob} and \eqref{eq:typicalcard} of the typical set can be expressed in terms of the typical projector $\Pi_\delta$:
\begin{gather}
  \label{eq:qtypicalprob}
  \tr(\Pi_\delta \rho_{\sA'}) \ge 1-\epsilon_\proj \eqp,
  \\
  \label{eq:qtypicalcard}
  \tr(\Pi_\delta) \le 2^{\kround(S+\delta)} \eqp.
\end{gather}
The first property can be used in the gentle operator lemma \cite{Win99,ON02} to show that a successful typicality measurement does not disturb the state by much \cite{Wil13}:
\begin{equation}
  \norm{\Pi_\delta \rho_{\sA'} \Pi_\delta - \rho_{\sA'}}_1
    \le 2 \sqrt{\epsilon_\pi} \eqp.
\end{equation}
Hence, the decompressed state is close in trace distance to the original \cite{Wil13}:
\begin{equation}
  \norm[\big]{
    V^\dag \mathcal C(\rho_{\sA'}) V
    - \rho_{\sA'}
  }_1
    \le 2\sqrt{\epsilon_\proj} + \epsilon_\proj
  \eqp.
\end{equation}

As Schumacher originally noted \cite{Sch95}, this remains true when we consider the global state in $\sA\sA'$; the entanglement of the state is therefore not destroyed by compression:
\begin{equation}
  \norm[\big]{
    (\id \otimes V^\dag)
      ({\idmap} \otimes \mathcal C) [\projop{\psi_\theta}^{\otimes \kround}]
      (\id \otimes V)
    - \projop{\psi_\theta}^{\otimes \kround}
  }_1
    \le 2\sqrt{\epsilon_\proj} + \epsilon_\proj
  \eqp.
\end{equation}

Bennett et al.'s dilution procedure simply results from the composition of this compression with quantum teleportation.
The proof of Lemma~\ref{lem:dilution} is therefore immediate:
\begin{proof}[Proof of Lemma~\ref{lem:dilution}]
Defining $
  \mathcal D_{\delta,\theta}(\projop{\phi^+}^{\otimes \kebit}_{\sA''\sB})
$ to be the outcome of the composition of a local preparation of $\projop{\psi_\theta}^{\otimes \kround}_{\sA\sA'}$, followed by Schumacher compression over $\delta$-typical sequences of $\sA'$, teleportation from $\sA'$ to $\sB$ using $\projop{\phi^+}^{\otimes \kebit}_{\sA\sB}$ and decompression on $\sB$, Lemma~\ref{lem:dilution} follows.
\end{proof}